\documentclass{style/vldb_style_sample/vldb}

\usepackage{graphicx}
\usepackage{balance}  % for  \balance command ON LAST PAGE  (only there!)

\usepackage{booktabs}
\usepackage{amsmath}
\usepackage{verbatim}
\usepackage{algorithmicx}
\usepackage[plain]{algorithm}
\usepackage{algpseudocode}
\usepackage{pgfplots}
\usepackage{pgfplotstable}
\usepackage[center]{subfigure}
\usepackage{url}
\usetikzlibrary{patterns,shapes,arrows,positioning,fit,decorations.pathreplacing}

\usepgfplotslibrary{groupplots}

\newlength{\abovecaptionskip}
\newfont{\mycrnotice}{ptmr8t at 7pt}
\newfont{\myconfname}{ptmri8t at 7pt}
\let\crnotice\mycrnotice%
\let\confname\myconfname%

\begin{document}

\newtheorem{lemma}{Lemma}
\newtheorem{theorem}{Theorem}
\newtheorem{corollary}{Corollary}
\newtheorem{example}{Example}
\newtheorem{assumption}{Assumption}
\newtheorem{remark}{Remark}

\newdef{definition}{Definition}
\newdef{scenario}{Scenario}

% ****************** TITLE ****************************************

\title{Solving the Join Ordering Problem \\via Mixed Integer Linear Programming}

% possible, but not really needed or used for PVLDB:
%\subtitle{[Extended Abstract]
%\titlenote{A full version of this paper is available as\textit{Author's Guide to Preparing ACM SIG Proceedings Using \LaTeX$2_\epsilon$\ and BibTeX} at \texttt{www.acm.org/eaddress.htm}}}

% ****************** AUTHORS **************************************

% You need the command \numberofauthors to handle the 'placement
% and alignment' of the authors beneath the title.
%
% For aesthetic reasons, we recommend 'three authors at a time'
% i.e. three 'name/affiliation blocks' be placed beneath the title.
%
% NOTE: You are NOT restricted in how many 'rows' of
% "name/affiliations" may appear. We just ask that you restrict
% the number of 'columns' to three.
%
% Because of the available 'opening page real-estate'
% we ask you to refrain from putting more than six authors
% (two rows with three columns) beneath the article title.
% More than six makes the first-page appear very cluttered indeed.
%
% Use the \alignauthor commands to handle the names
% and affiliations for an 'aesthetic maximum' of six authors.
% Add names, affiliations, addresses for
% the seventh etc. author(s) as the argument for the
% \additionalauthors command.
% These 'additional authors' will be output/set for you
% without further effort on your part as the last section in
% the body of your article BEFORE References or any Appendices.

\numberofauthors{1} %  in this sample file, there are a *total*
% of EIGHT authors. SIX appear on the 'first-page' (for formatting
% reasons) and the remaining two appear in the \additionalauthors section.

\author{
% You can go ahead and credit any number of authors here,
% e.g. one 'row of three' or two rows (consisting of one row of three
% and a second row of one, two or three).
%
% The command \alignauthor (no curly braces needed) should
% precede each author name, affiliation/snail-mail address and
% e-mail address. Additionally, tag each line of
% affiliation/address with \affaddr, and tag the
% e-mail address with \email.
%
% 1st. author
\alignauthor
Immanuel Trummer and Christoph Koch\\
			 \email{\{firstname\}.\{lastname\}@epfl.ch}\\
       \affaddr{\'Ecole Polytechnique F\'ed\'erale de Lausanne}
}

\maketitle

% introduce new font
\newcommand*{\codeF}{\fontfamily{\sfdefault}\selectfont}
\newcommand*{\branchF}{\fontsize{9}{9}\selectfont}
\newcommand*{\connectorF}{\fontsize{9}{9}\selectfont}

\begin{abstract}
%Many database-related optimization problems are solved by transforming them into mixed integer linear programs (MILP) as this allows to apply existing solvers which have reached a high level of maturity. 
We transform join ordering into a mixed integer linear program (MILP). This allows to address query optimization by mature MILP solver implementations that have evolved over decades and steadily improved their performance. They offer features such as anytime optimization and parallel search that are highly relevant for query optimization. %provide useful features such as  Connecting query optimization to MILP enables us to benefit from decades of research in MILP and corresponding solver implementations provide useful features such as anytime behavior and parallelization.%Query optimizers are usually expensive to develop and to maintain; our approach allows to replace the optimizer core by a MILP solver which significantly reduces the size of the optimizer code base. Furthermore, 

We present a MILP formulation for searching left-deep query plans. We use sets of binary variables to represent join operands and intermediate results, operator implementation choices or the presence of interesting orders. Linear constraints restrict value assignments to the ones representing valid query plans. We approximate the cost of scan and join operations via linear functions, allowing to increase approximation precision up to arbitrary degrees. Our experimental results are encouraging: we are able to find optimal plans for joins of 60 tables; a query size that is beyond the capabilities of prior exhaustive query optimization methods.
\end{abstract}

%\pgfplotscreateplotcyclelist{exampleCL}{{mark=square*, mark size=2, black, draw=black, solid},{mark=*, mark size=2, black, draw=black, dashed},{mark=triangle*, mark size=3, black, draw=black, dotted}}

\section{Introduction}
\label{introSec}

From the developer's perspective, there are two ways of solving a hard optimization problem on a computer: either we write optimization code from scratch that is customized for the problem at hand or we transform the problem into a popular problem formalism and use existing solver implementations.
In principle, the first approach could lead to more efficient code as it allows to exploit specific problem properties. Also, we do not require a transformation that might blow up the size of the problem representation. In practice however, our customized code competes against mature solver implementations for popular problem models that have been fine-tuned over decades~\cite{Bixby2012}, driven by a multitude of application scenarios. Using an existing solver reduces the amount of code that needs to be written and we might obtain desirable features such as parallel optimization or anytime behavior (i.e., obtaining solutions of increasing quality as optimization progresses) automatically from the solver implementation. It is therefore in general advised to consider and to evaluate both approaches for solving an optimization problem.

We apply this generic insight to the problem of database query optimization. For the last thirty years, the problem of exhaustive query optimization, more precisely the core problem of join ordering and operator selection~\cite{Selinger1979}, has typically been solved by customized code inside the query optimizer. Query optimizers consist of millions of code lines~\cite{Waas2009} and are the result of thousands of man years worth of work~\cite{Kaushik2009}. The question arises whether this development effort is actually necessary or whether we can transform query optimization into another popular problem formalisms and use existing solvers. We study that question in this paper.

We transform the join ordering problem into a mixed integer linear program (MILP). We select that formalism for its popularity. Integer programming approaches are currently the method of choice to solve thousands of optimization problems from a wide range of areas~\cite{Lawrence1997}. Corresponding software solvers have sometimes evolved over decades and reached a high level of maturity~\cite{Bixby2012}. Commercial solvers such as Cplex\footnote{\url{http://www.ibm.com/software/products/en/ibmilogcpleoptistud}} or Gurobi\footnote{\url{http://www.gurobi.com/}} are available for MILP as well as open source alternatives such as SCIP\footnote{\url{http://scip.zib.de/}}. 

Those solvers offer several features that are useful for query optimization. First of all, they possess the anytime property: they produce solutions of increasing quality as optimization progresses and are able to provide bounds for how far the current solution is from the optimum. Chaudhuri recently mentioned the development of anytime algorithms as one of the relevant research challenges in query optimization~\cite{Chaudhuri2009}. Mapping query optimization to MILP immediately yields an algorithm with that property (note that recently proposed anytime algorithms for multi-objective query optimization~\cite{Trummer2015a} are not applicable to traditional query optimization). Second, MILP solvers already offer support for parallel optimization which is an active topic of research in query optimization as well~\cite{Han2008, Waas2009, Soliman2014}. Finally, the performance of MILP solvers has improved (hardware-independently) by more than factor 450,000 over the past twenty years~\cite{Bixby2012}. It seems entirely likely that those advances can speed up query optimization as well (and anticipating our experimental results, we find indeed classes of query optimization problems where a MILP based approach treats query sizes that are illusory for prior exhaustive query optimization algorithms). 

In summary, by connecting query optimization to integer programming, we benefit from over sixty years of theoretical research and decades of implementation efforts. Even better, having a mapping from query optimization to MILP does not only enable us to benefit from past research but also from all future research and development advances that originate in the fruitful area of MILP. Performance improvements have been steady in the past~\cite{Bixby2012} and, as several major software vendors compete in that market, are likely in the future as well.

Given that integer programming transformations have been proposed for many optimization problems that connect to query optimization~\cite{Agarwal2013, Beame2014, Dokeroglu, Papadomanolakis2007, Yang1997}, it is actually surprising that no such mapping has been proposed for the join ordering problem itself so far. There are even sub-domains of query optimization, notably parametric query optimization~\cite{Ganguly1998, Hulgeri2002, Hulgeri2003} and multi-objective parametric query optimization~\cite{Trummer2015}, where it is common to approximate the cost of query plans via piecewise-linear functions. The purpose here is however to model the dependency of plan cost on unknown parameters while traditional approaches such as dynamic programming are used to find the optimal join order. None of the aforementioned publications transforms the join ordering problem into a MILP and the same applies for additional related work that we discuss in Section~\ref{relatedSec}.

%Example problems for which integer programming formulations have been proposed include for instance multiple query optimization~\cite{Dokeroglu}, index selection~\cite{Papadomanolakis2007}, materialized view design~\cite{Yang1997}, selection of data samples~\cite{Agarwal2013}, or partitioning of data for parallel processing~\cite{Beame2014}. None of the aforementioned transformations is however applicable to the classical problem of join ordering and operator selection. Many publications in the domain of parametric query optimization or multi-objective parametric query optimization model the cost of query plans by linear or piecewise-linear functions~\cite{Ganguly1998, Hulgeri2002, Hulgeri2003, Trummer2015}. Their goal is however not to replace the optimizer core by a generic solver but rather to model the dependency of plan cost on unknown parameters. Some of the aforementioned approaches use linear programming in the pruning sub-function of a dynamic programming-based solver implementation~\cite{Hulgeri2002, Trummer2015} while others calculate interesting parameter value combinations for which a standard query optimizer is invoked based on linear constraints~\cite{Hulgeri2003}. In each case the optimal join order is not decided by a generic solver alone.

%  (e.g., the goal of multiple query optimization is to select one out of a set of previously generated query plans for each query but not to generate the query plans themselves~\cite{Dokeroglu})

A MILP is specified by a set of variables with either continuous or integer value domain, a set of linear constraints on those variables, and a linear objective function that needs to be minimized. An optimal solution to a MILP is an assignment from variables to values that minimizes the objective function. We sketch out next how we transform the join ordering problem into a MILP.

Left-deep query plans can be represented as follows (we simplify by not considering alternative operator implementations while the extensions are discussed later). For a given query, we can derive the total number of required join operations from the number of query tables. As we know the number of required joins in advance, we introduce for each join operand and for each query table a binary variable indicating whether the table is part of that join operand. We add linear constraints enforcing for instance that single tables are selected for the inner join operands (a particularity of left-deep query plans), that the outer join operands are the result of the prior join (except for the first join), or that join operands have no overlap. The result is a MILP where each solution represents a valid left-deep query plan.

%We select the MILP formalism because of its popularity. Thousands of optimization problems from a wide range of areas~\cite{Lawrence1997}, inside and outside the database community, are solved via a transformation into that formalism. Commercial software solvers for MILP problems are available (e.g., IBM Cplex\footnote{\url{http://www-03.ibm.com/software/products/en/ibmilogcpleoptistud}} or Gurobi\footnote{\url{http://www.gurobi.com/}}) that have evolved over many years and have reached a high degree of sophistication. The same applies to various solver alternatives from the open source domain (e.g., GLPK\footnote{\url{http://www.gnu.org/software/glpk/}} or LP solve\footnote{\url{http://lpsolve.sourceforge.net/5.5/}}). 

This is not yet useful: we must associate query plans with cost in order to obtain the optimal plan from the MILP solver. The cost of a query plan depends on the cardinality (or byte size) of intermediate results. The cardinality of an intermediate result depends on the selected tables and on the evaluated predicates. We introduce a binary variable for each predicate and each intermediate result, indicating whether the predicate has been evaluated to reduce cardinality. Predicate variables are restricted by linear constraints that make it impossible to evaluate a predicate as long as not all query tables it refers to are present in the corresponding result. The cardinality of the join of a set of tables on which predicates have been evaluated is usually estimated by the product of table cardinalities and predicate selectivities. As we cannot directly represent a product via linear constraints, we focus on the logarithm of the cardinality: the logarithm of a product is the sum of the logarithms of the factors. Based on our binary variables representing selected tables and evaluated predicates, we calculate the logarithm of the cardinality for all intermediate results that appear in a query plan. Based on the logarithm of the cardinality, we approximate the cost of query plans via sets of linear constraints and via auxiliary variables. 

We must approximate cost functions since the cost of standard operators is usually not linear in the logarithm of input and output cardinalities. We can however choose the approximation precision by choosing the number of constraints and auxiliary variables. This allows in principle arbitrary degrees of precision. Also note that there are entire sub-domains of query optimization in which it is standard to approximate plan cost functions via linear functions~\cite{Ganguly1998, Hulgeri2002, Hulgeri2003, Trummer2015}. Approximating plan cost via linear function is therefore a widely-used approach.

%So far we have discussed our representation of left-deep plans. We use the same variables to represent join operands. As the choice of join operands is however more flexible for bushy plans, the constraints that are needed to exclude the inadmissible operand choices actually become more complicated. We must introduce an additional binary variable for each pair of tables and for each join operation indicating whether or not the two tables are joined together after the join operation (either as a consequence of the current join or of a prior join). Based on those variables, we can for instance enforce that join operands are either single tables or table sets that have been previously joined. The cost of bushy query plans can be calculated using similar techniques as for left-deep plans.

Our goal here was to give a first intuition for how our transformation works and we have therefore considered join order alone and in a simplified setting. Later we show how to extend our approach for representing alternative operator implementations, complex cost models taking into account interesting orders and the evaluation cost of expensive predicates, or richer query languages.

We formally analyze our transformation in terms of the resulting number of constraints and variables. In our experimental evaluation, we apply the Gurobi MILP solver to query optimization problems that have been reformulated as MILP problems. We compare against a classical dynamic programming based query optimization algorithm on different query sizes and join graph structures. Our results are encouraging: the MILP approach often generates guaranteed near-optimal query plans after few seconds where dynamic programming based optimization does not generate any plans up to the timeout of one minute. 

The original scientific contributions of this paper are the following:

\begin{itemize}
\item We show how to reformulate query optimization as MILP problem.
\item We analyze the problem mapping and express the number of variables and constraints as function of the query dimensions.
\item We evaluate our approach experimentally and compare against a classical dynamic programming based query optimizer.
\end{itemize}

The remainder of this paper is organized as follows. We discuss related work in Section~\ref{relatedSec}. In Section~\ref{modelSec}, we introduce our formal problem model. Section~\ref{approachSec} describes how we transform query optimization into MILP. We analyze how the size of the resulting MILP problem grows in the dimension of the original query optimization problem in Section~\ref{analysisSec}. In Section~\ref{experimentsSec}, we experimentally evaluate an implementation of our MILP approach in comparison with a classical dynamic programming based query optimization algorithm.

\newpage
\section{Related Work}
\label{relatedSec}

MILP representations have been proposed for many optimization problems in the database domain, including but not limited to multiple query optimization~\cite{Dokeroglu}, index selection~\cite{Papadomanolakis2007}, materialized view design~\cite{Yang1997}, selection of data samples~\cite{Agarwal2013}, or partitioning of data for parallel processing~\cite{Beame2014}. In the areas of parametric query optimization and multi-objective parametric query optimization it is common to model the cost of query plans by linear functions that depend on unknown parameters~\cite{Ganguly1998, Hulgeri2002, Hulgeri2003, Trummer2015}. None of those prior publications formalizes however the join ordering and operator selection problem as MILP. 

Query optimization algorithms can be roughly classified into exhaustive algorithms that formally guarantee to find optimal query plans and into heuristic algorithms which do not possess those formal guarantees. Exhaustive query optimization algorithms are often based on dynamic programming~\cite{Selinger1979, Vance1996a, Moerkotte2006, Moerkotte2008}. We compare against such an approach in our experimental evaluation. 

Our MILP-based approach to query optimization can be used as an exhaustive query optimization algorithm since we can configure the MILP solver to return a guaranteed-optimal solution. The MILP solver can however easily be configured to return solutions that are guaranteed near-optimal (i.e., the cost of the result plan is within a certain factor of the optimum) or to return the best possible plan within a given amount of time. This makes the MILP approach more flexible than typical exhaustive query optimization algorithms. Furthermore, MILP solvers posses the anytime property, meaning that they produce multiple plans of decreasing cost during optimization. The development of anytime algorithms for query optimization has recently been identified as a research challenge~\cite{Chaudhuri2009}. Transforming query optimization into MILP immediately yields anytime query optimization. Note that anytime algorithms for multi-objective query optimization~\cite{Trummer2015a} cannot speed up traditional query optimization with one plan cost metric.

The parallelization of exhaustive query optimization algorithms (not to be confused with query optimization for parallel execution) is currently an active research topic~\cite{Han2008, Han2009, Soliman2014, Waas2009}. MILP solvers such as Cplex or Gurobi are able to exploit parallelism and transforming query optimization into MILP hence yields parallel query optimization as well. The development of parallel query optimizers for new database systems requires generally significant investments~\cite{Soliman2014}; the amount of code to be written can be significantly reduced by using a generic solver as optimizer core. 

Various heuristic and randomized algorithms have been proposed for query optimization~\cite{Bennett1991, Bruno, ioannidis1990randomized, Steinbrunn1997, Swami1988, Swami1989}. In contrast to many exhaustive algorithms, most of them possess the anytime property and generate plans of improving quality as optimization progresses. Those approaches can however not give any formal guarantees at any point in time about how far the current solution is from the optimum. MILP solvers provide upper-bounds during optimization on the cost difference between the cost of the current solution and the theoretical optimum. Such bounds can for instance be used to stop optimization once the distance reaches a threshold. Randomized algorithms do not offer that possibility and the returned solutions may be arbitrarily far from the optimum.

\section{Model and Assumptions}
\label{modelSec}

The goal of query optimization is to find an optimal or near-optimal plan for a given query. It is common to introduce new query optimization algorithms by means of simplified problem models. We also use a simple query and query plan model throughout most of the paper while we discuss extensions to richer query languages and plan models as well. 

In our simplified model, we represent a query as a set $Q$ of tables that need to be joined together with a set $P$ of binary predicates that connect the tables in $Q$ (extensions to nested queries, queries with aggregates, queries with projections, and queries with non-binary predicates will be discussed). For each binary predicate $p\in P$, we designate by $T_1(p),T_2(p)\in Q$ the two tables that the predicate refers to. Predicates can only be evaluated in relations in which both tables they refer to have been joined.

We assume in the simplified problem model that one scan and one binary join operator are available. As we consider binary joins, a query with $n$ tables requires $n-1$ join operations. A query plan is defined by the operands of those $n-1$ join operations, more precisely by the tables that are present in those operands. We consider left-deep plans. For left-deep query plans, the inner operand is always a single table; the outer operand is the result from the previous join (except for the outer operand of the first join which is a single table).

Query plans are compared according to their execution cost. The execution cost of a plan depends on the cardinality of the intermediate results it produces. We write $Card(t)\geq 1$ to designate the cardinality of table $t$ and $Sel(p)\in(0,1]$ to designate the selectivity of predicate $p$. We assume in the simplified model that the cardinality of the join between several tables, after having evaluated a set of join predicates, corresponds to the product of the table cardinalities and the predicate selectivities. We hence assume in the simplified model uncorrelated predicates while extensions to correlated predicates will be discussed. We generally assume that the execution cost of a query plan is the sum of the execution cost of all its operations. We will show how to represent various cost functions. 

We translate the problem of finding a cost-minimal plan for a given query into a mixed integer linear programming problem (MILP). A MILP problem is defined by a set of variables (that can have either integer or continuous value domains), a set of linear constraints on those variables, and a linear objective function on those variables that needs to be minimized. A solution to a MILP is an assignment from variables to values from the respective domain such that all constraints are satisfied. An optimal solution minimizes the objective function value among all solutions.

\section{Join Ordering Approach}
\label{approachSec}

%TODO: minimizing sum of intermediate result sizes minimizes many other classical join cost functions as well~\cite{Cluet1995}; selecting join operators afterwards: this is an easy problem while the hard problem (join ordering) is solved by the MILP solver

The join ordering problem is usually solved by algorithms that are specialized for that problem and run inside the query optimizer. We adopt a radically different approach: we translate the join ordering problem into a MILP problem that we solve by a generic MILP solver.

MILP is an extremely popular formalism that is used to solve a variety of problems inside and outside the database community. By mapping the join ordering problem into a MILP formulation, we benefit from decades of theoretical research in the area of MILP as well as from solver implementations that have reached a high level of maturity. By linking query optimization to MILP, we make sure that query optimization will from now on indirectly benefit from all theoretical advances and refined implementations that become available in the MILP domain.

We explain in the following our mapping from a join ordering problem to a MILP. We describe the variables and constraints by which we represent valid join orders in Section~\ref{joinOrderSub}. We show how to model the cardinality of join operands in Section~\ref{cardinalitySub}. In Section~\ref{costSub} we associate plans with cost values based on the operand cardinalities.

Note that we introduce our mapping by means of a basic problem model in this section while we discuss extensions to the query language, plan space, and cost model in Section~\ref{extensionsSec}.

\subsection{Join Order}
\label{joinOrderSub}

A MILP program is characterized by variables with associated value domains, a set of linear constraints on those variables, and a linear objective function on those variables that needs to be minimized. Table~\ref{varsLinearTable} summarizes the variables that we require to model join ordering as MILP problem and Table~\ref{leftDeepConstraints} summarizes the associated constraints. We introduce them step-by-step in the following.

\begin{table*}[t!]
\centering
\caption{Variables for formalizing join ordering for left-deep query plans as integer linear program.\label{varsLinearTable}}
\begin{tabular}{lll}
\toprule[1pt]
\textbf{Symbol} & \textbf{Domain} & \textbf{Semantic}\\
\midrule[1pt]
$tio_{tj}$/$tii_{tj}$ & $\{0,1\}$ & If table $t$ is in outer/inner operand of $j$-th join\\
\midrule
$pao_{pj}$ & $\{0,1\}$ & If $p$-th predicate can be evaluated on outer operand of $j$-th join\\
\midrule
$lco_j$ & $\mathbb{R}$ & Logarithm of cardinality of outer operand of $j$-th join\\
\midrule
$cto_{rj}$ & $\{0,1\}$ & If cardinality of outer operand of $j$-th join reaches $r$-th threshold\\
\midrule
$co_j$/$ci_j$ & $\mathbb{R}_{+}$ & Approximated cardinality of outer/inner operand of $j$-th join\\
\bottomrule[1pt]
\end{tabular}
\end{table*}

\begin{table*}[t!]
\centering
\caption{Constraints for join ordering in left-deep plan spaces.\label{leftDeepConstraints}}
\begin{tabular}{lll}
\toprule[1pt]
\textbf{Constraint} & \textbf{Semantic}\\
\midrule[1pt]
$\sum_t tio_{t0}=1$/$\forall j:\sum_t tii_{tj}=1$ & Select one table for outer operand of first join/for all inner operands\\
\midrule
$\forall j\forall t:tio_{tj}+tii_{tj}\leq 1$ & The tables in the join operands cannot overlap for the same join\\
\midrule
$\forall j\geq1\forall t:tio_{tj}=tii_{t,j-1}+tio_{t,j-1}$ & Results of prior join are outer operand for next join\\
\midrule
$\forall p\forall j:pao_{pj}\leq tio_{T_1(p)j};pao_{pj}\leq tio_{T_2(p)j}$ & Predicates are applicable if both referenced tables are in outer operand\\
\midrule
$\forall j:ci_{j}=\sum_t Card(t)tii_{tj}$ & Determines cardinality of inner operand\\
\midrule
$\forall j:lco_j=\sum_t\log(Card(t))tio_{tj}+$& Determines logarithm of outer operand cardinality,\\
$\quad\sum_p\log(Sel(p))pao_{pj}$ & taking into account selected tables and applicable predicates\\
\midrule
$\forall j\forall r:lco_{j}-cto_{rj}\cdot\infty\leq\log(\theta_r)$& Activates threshold flag if cardinality reaches threshold\\
\midrule
$\forall j:co_j=\sum_r cto_{rj}\delta \theta_r$ & Translates activated thresholds into approximate cardinality \\
\bottomrule[1pt]
\end{tabular}
\end{table*}

We start by discussing the variables and constraints that we need in order to represent valid left-deep query plans. Later we discuss the variables and constraints that are required to estimate the cost of query plans.

We represent left-deep query plans for a query $Q$ as follows. For the moment, we assume that only one join operator and one scan operator are available while we discuss extensions in Section~\ref{extensionsSec}. Under those assumptions, a query plan is specified by the join operands. We introduce a set of binary variables $tio_{tj}$ (short for \textit{Table In Outer join operand}) with the semantic that $tio_{tj}$ is one if and only if query table $t\in Q$ appears in the outer join operand of the $j$-th join. We numerate joins from 0 to $j_{max}$ where $j_{max}$ is determined by the number of query tables. Analogue to that, we introduce a set of binary variables $tii_{tj}$ (short for \textit{Table In Inner join operand}) indicating whether the corresponding table is in the inner operand of the $j$-th join. 

The variables representing left-deep plans have binary value domains. Note that not all possible value combinations represent a valid left-deep plan. For instance, we could represent joins with empty join operands. Or we could build plans that join only a subset of the query tables and are therefore incomplete. We must impose constraints in order to restrict the considered value combinations to the ones representing valid and complete left-deep plans. 

Left-deep plans are characterized by the particularity that the inner operand consists of only one table for each join. We capture that fact by the constraint $\sum_t tii_{tj}=1$ which we need to introduce for each join~$j$. A similar constraint restricts the table selections for the outer operand of the first join (join index $j=0$) as only one table can be selected as initial operand. For the following joins (join index $j\geq1$), the outer join operand is always the result of the previous join which is another characteristic of left-deep plans. This translates into the constraints $tio_{tj}=tii_{t,j-1}+tio_{t,j-1}$.

The latter constraint actually excludes the possibility that the same table appears in both operands of a join (since the result of the sum between $tii_{t,j-1}+tio_{t,j-1}$ cannot exceed the maximal value of one for $tio_{tj}$) except for the last join. We add the constraint $tio_{tj_{max}}+tii_{tj_{max}}\leq 1$ for the last join (and optionally for the other joins as well). 

The number of joins is one less than the number of query tables. We join two (different) tables in the first join. After that, each join adds one new table to the set of joined tables since the outer operand contains all tables that have been joined so far, since the inner operand consists of one table, and since inner and outer join operands do not overlap. As a result, we can only represent complete query plans that join all tables. 

We could have chosen a different representation of query plans with less variables. The problem is that we need to be able to approximate the cost of query plans based on that representation using linear functions. Our representation of query plans might at first seem unnecessarily redundant but it allows to impose the constraints that we discuss next. Also note that MILP solvers typically try to eliminate unnecessary variables and constraints in preprocessing steps. This makes it less important to reduce the number of variables and constraints at the cost of readability. 

\begin{example}
We illustrate the representation of left-deep query plans for the join query $R\Join S\Join T$. Answering the query requires two join operations. Hence we introduce six variables $tio_{tj}$ for $t\in\{R,S,T\}$ and $j\in\{0,1\}$ to represent outer join operands and six variables $tii_{tj}$ to represent inner join operands. The join order $(R\Join S)\Join T$ is for instance represented by setting $tio_{R0}=tii_{S0}=1$ and $tio_{R1}=tio_{S1}=tii_{T1}=1$ and setting the other variables representing join operands to zero. This assignment satisfies the two constraints that restrict inner operands to single tables (e.g., $\sum_{t\in\{R,S,T\}}tii_{t1}=1$ for the second join), it satisfies the constraint restricting the outer operand in the first join to a single table ($\sum_{t\in\{R,S,T\}}tio_{t0}=1$), and it satisfies the constraints making the outer operand of the second join equal to the union of the operands in the first join (e.g., $tio_{R1}=tio_{R0}+tii_{R0}$).
\end{example}

\subsection{Cardinality}
\label{cardinalitySub}

Our goal is to find query plans with minimal cost and hence we must associate query plans with a cost value. The execution cost of a query plan depends heavily on the cardinality of intermediate results. We need to represent the cardinality of join operands and join results in order to calculate the cost of query plans. Inner operands consist always of a single table and calculating their cardinality is straight-forward: designating by $ci_{j}$ (short for \textit{Cardinality of Inner operand}) the cardinality of the inner operand of join number~$j$, we simply set $ci_j=\sum_t tii_{tj}Card(t)$ where $Card(t)$ is the cardinality of table $t$.

Calculating cardinality for outer join operands is however non-trivial as we can only use linear constraints: the cardinality of a join result is usually estimated as the product of the cardinalities of the join operands times the selectivity of all predicates that are applied during the join. The product is a non-linear function and does not directly translate into linear constraints. 

We circumvent that problem via the following trick. While cardinality is actually defined as the product of table cardinality values and predicate selectivity values, we represent the logarithm of the cardinality instead and the logarithm of a product is the sum of the logarithms of the factors. More formally, given a set $T\subseteq Q$ of query tables such that the set of predicates $P$ is applicable to $T$ (i.e., for each binary predicate in $P$ the two tables it refers to are included in $T$) and designating by $Card(t)$ for $t\in T$ the cardinality of table $t$ and by $Sel(p)$ the selectivity of predicate $p\in P$, a cardinality estimate is given by $\prod_{t\in T}Card(t)\cdot\prod_{p\in P}Sel(p)$ and the logarithm of the cardinality estimate is $\sum_{t\in T}\log(Card(t))+\sum_{p\in P}\log(Sel(p))$ which is a linear function. 

We introduce the set of variables $lco_{j}$ (short for \textit{Logarithmic Cardinality of Outer operand}) which represents the logarithm of the cardinality of the outer operand of the $j$-th join. The aforementioned linear formula for calculating the logarithm of the cardinality depends on the selected tables as well as on the applicable predicates. The selected tables are directly given in the variables $tio_{tj}$. We introduce additional binary variables to represent the applicable predicates: variable $pao_{pj}$ (short for \textit{Predicate Applicable in Outer join operand}) captures whether predicate $p$ is applicable in the outer operand of the $j$-th join. We currently consider only binary predicates (we discuss extensions later) and as the inner operands consist of single tables, we do not need to introduce an analogue set of predicate variables for the inner operands.

We denote by $T_1(p)$ and $T_2(p)$ the first and the second table that predicate $p$ refers to. A predicate is applicable to an operand whose table set $T$ contains $T_1(p)$ and  $T_2(p)$. We make sure that predicates cannot be applied if one of the two tables is missing by adding for each predicate $p$ and each join $j$ a pair of constraints of the form $pao_{pj}\leq tio_{T_1(p)}$ and $pao_{pj}\leq tio_{T_2(p)}$. We currently assume that predicate evaluations do not incur any cost while extensions are discussed later. Under this assumption, applying a predicate has only beneficial effects as it reduces the cardinality of intermediate results and therefore the cost of the following joins. This means that we only need to introduce constraints preventing the solver from using predicates that are inapplicable but we do not need to add constraints forcing the evaluation of predicates explicitly.

Using the variables capturing the applicability of predicates, we can now write the logarithm of the join operand cardinalities. For outer join operands, we set \[lco_j=\sum_t\log(Card(t))tio_{tj}+\sum_p\log(Sel(p))pao_{pj}\] and thereby take into account table cardinalities as well as predicate selectivities. 

Unfortunately, the cost of most operations within a query plan is not linear in the logarithm of the cardinality values. In the following, we show how to transform the logarithm of the cardinality values into an approximation of the raw cardinality values. This allows to write cost functions that are linear in the cardinality of their input and output. This is sufficient for many but not for all standard operations. Similar techniques to the ones we describe in the following can however be used to represent for instance log-linear cost functions as we describe in more detail in Section~\ref{costSub}.

We must transform the logarithm of the cardinality into the cardinality itself. This is not a linear transformation and hence we resort to approximation. We assume that a set $\Theta=\{\theta_r\}$ of cardinality threshold values has been defined for integer indices $r$ with $0\leq r\leq r_{max}$. In addition, we introduce a set of binary variables $cto_{rj}$ (short for \textit{Cardinality Threshold reached by Outer operand}) that indicate for each join~$j$ and each cardinality threshold value $\theta_r$ whether the cardinality of the outer operand reaches the corresponding threshold value. If threshold $\theta_r$ is reached then the corresponding threshold variable $cto_{rj}$ must take value one and otherwise value zero. To guarantee that the previous statement holds, we introduce constraints of the form $lco_{j}-cto_{rj}\cdot\infty\leq\log(\theta_r)$ for each join~$j$ where $\infty$ is in practice a sufficiently large constant such that the constraint can be satisfied by setting the threshold variable $cto_{rj}$ to one. We do not explicitly enforce that the threshold variable is set to zero in case that the threshold is not reached. The constraints that we introduce next make however sure that the cardinality estimate and therefore the cost estimate increase with every threshold variable that is set to one. Hence the solver will set the threshold variables to zero wherever it can.

Based on the threshold variables, we can formulate a linear approximation for the raw cardinality. We introduce the set of variables $co_{j}$ representing the raw cardinality of the outer operand of the $j$-th join and set $co_j=\sum_r cto_{rj}\delta \theta_r$ where the values $\delta \theta_r$ are chosen appropriately such that if threshold variables $cto_{0j}$ up to $cto_{mj}$ are set to one for some specific join $j$ then the cardinality variable $co_j$ takes a value between $\theta_m$ and $\theta_{m+1}$ (assuming that thresholds are indexed in ascending order such that $\forall r:\theta_r<\theta_{r+1}$). We can for instance set $\delta\theta_r=\theta_r-\theta_{r-1}$ for $r\geq 1$ and $\delta\theta_0=\theta_0$.

\begin{example}
We illustrate how to calculate join operand cardinalities and continue the previous example with join query $R\Join S\Join T$. We have two joins and introduce therefore four variables ($ci_0$, $ci_1$, $co_0$, and $co_1$) representing operand cardinalities. Assume that tables $R$, $S$, and $T$ have cardinalities 10, 1000, and 100 respectively. We calculate the cardinality of the two inner join operands by summing over the variables indicating the presence of a table in an inner operand, weighted by the cardinality values (e.g., $ci_0=10tii_{R0}+1000tii_{S0}+100tii_{T0}$). The cardinality of the outer operands can depend on predicates. Assume that one predicate $p$ is defined between tables $R$ and $S$. We introduce two variables, $pao_{p0}$ and $pao_{p1}$, indicating whether the predicate can be evaluated in the outer operand of the corresponding join. Predicates can be evaluated if both referenced tables are in the corresponding operand. We introduce four constraints (e.g., $pao_{p0}\leq tio_{R0}$ and $pao_{p0}\leq tio_{S0}$) forcing the value of the predicate variable to zero if at least one of the tables is not present. We introduce two variables storing the logarithm of the outer operand cardinality: $lco_0$ and $lco_1$. We assume that the selectivity of $p$ is 0.1. Then the logarithmic cardinality for the first outer join operand is given by $lco_0=1pao_{R0}+3pao_{S0}+2pao_{T0}-1pao_{p0}$, assuming that the logarithm base is 10. To simplify the example, we assume that only two cardinality thresholds are considered: $\theta_0=10$, and $\theta_1=1000$. We introduce four variables $cto_{rj}$ with $r\in \{0,1\}$ and $j\in\{0,1\}$ indicating whether the cardinality of the outer join operand reaches each threshold for the first or second join. Each threshold variable is constrained by one constraint (e.g., $lco_0-\infty\cdot cto_{0,0}\leq 1$). Now we define the cardinality of the outer join operands by constraints such as $co_0=10cto_{0,0}+(1000-10)cto_{1,0}$. This provides a lower bound for the true cardinality. If we know for instance that cardinality values are upper-bounded by 100000 due to the query properties, we can also set $co_0=100cto_{0,0}+(10000-100)cto_{1,0}$. Then the difference between true and approximate cardinality is at most one order of magnitude.
\end{example}

\subsection{Cost}
\label{costSub}

Now we can for instance sum up the cardinalities over all intermediate results ($\sum_{j\geq1}cio_{j}$) and thereby obtain a simple cost metric that is equivalent to the $C_{out}$ cost metric introduced by Cluet and Moerkotte~\cite{Cluet1995}. Join orders minimizing that cost metric were shown to minimize cost according to the cost formulas of some of the standard join operators as well~\cite{Cluet1995}. We will however show in the following how the cost of all standard join operators, namely hash join, sort-merge join, and block nested loop join, can be modeled directly. 

The standard cost formula for a hash join operation is based on the number of pages that the two input operands consume on disk. We designate by $pgo_j$ the number of disk pages consumed by the outer operand of join number $j$ and $pgi_j$ is the analogue value for the inner operand. If a hash join operator is used for the join then its cost is given by $3\cdot(pgo_j+pgi_j)$. This is a linear formula but we must calculate the size of the operands in disc pages. 

The byte size of an intermediate result, and therefore the number of consumed disk pages, depends not only on the cardinality but also on the columns that are present. For the moment, we make the simplifying assumption that each tuple has a fixed byte size. We show how to relax that restriction in the next section. Under this simplifying assumption, we can however express the disk pages of the outer operands as $pgo_j=\lceil co_j\cdot tupSize/pageSize\rceil$ where $tupSize$ is the fixed byte size per tuple and $pageSize$ the number of bytes per disk page. Factor $tupSize/pageSize$ is a constant due to our simplifying assumption and hence we can set $pgo_j=co_j\cdot tupSize/pageSize$ to obtain the approximate number of disk pages. Alternatively, we could write $pgo_j=\sum_{r}\lceil \theta_r\cdot tupSize/pageSize\rceil (cto_{jr}-cto_{j,r+1})$ and approximate it using the threshold variables (the expression $(cto_{jr}-cto_{j,r+1})$ yields value one only for the threshold variable with the highest threshold that is still set to one). Note that the factors of the form $\lceil \theta_r\cdot tupSize/pageSize\rceil$ are constants. The second version has the advantage that we can explicitly control the approximation precision for $pgo_j$ by tuning the number of thresholds. The disc pages for the inner operands can be obtained in a simplified way as each inner operand consists of only one table: we simply set $pgi_j=\sum_t tii_{tj}\lceil Card(t)\cdot tupSize/pageSize\rceil$.

The cost of sort-merge join operators can be approximated in a similar way. We assume here that both inputs must be sorted while we generalize in the next subsection. If both input operands need to be sorted first then the join cost is given by $2pgo_j\lceil\log(pgo_j)\rceil+2pgi_j\lceil\log(pgi_j)\rceil+pgo_j+pgi_j$. We have already shown how to obtain the number of disc pages $pgo_j$ and $pgi_j$. The log-linear numbers of disc pages, $pgo_j\log(pgo_j)$ and $pgi_j\log(pgi_j)$, can be obtained in a similar way. We use the cardinality thresholds for the outer operand and simply sum over tables for the inner operand. 

The cost function for the block nested loop join is given by $\lceil pgo_j/buffer\rceil\cdot pgi_j$ where $buffer$ is the amount of buffer space dedicated to the outer operand. We assume here that pipelining is used while the generalization is straightforward. There are several options for approximating that cost function with linear constraints. We can approximate the join cost function by omitting the ceiling operator and obtain $pgo_j/buffer\cdot pgi_j$. Similar to how we calculated the cardinality of the outer operands, we can switch to a logarithmic representation and write the logarithm of the join cost as $\log(pgo_j)+\log(pgi_j)-\log(buffer)$. Then we can transform the logarithm of the join cost into the raw join cost value using a set of newly introduced threshold variables. 

Another idea is to exploit the specific shape of the inner join operands. As only one table is selected for the inner join operand, we can express join cost by the formula $\sum_t tii_{tj}\cdot pages(t)\cdot blocks_j$ where $pages(t)=\lceil Card(t)\cdot tupSize/pageSize\rceil$ designates the disk page size of table $t$ and $blocks_j=\lceil pgo_j/buffer\rceil\approx pgo_j/buffer$ is the number of iterations of the outer loop executed by the block nested loop join. This is a weighted sum over products between a binary variable (the variables $tii_{tj}$ indicating whether table $t$ was selected for the inner operand of join number $j$) and a continuous variable (the variables $blocks_j$). This formula is hence not directly linear but the product between a binary variable and a continuous variable can be expressed by introducing one auxiliary variable and a set of constraints~\cite{Modeling2015}. The only condition for this transformation is that the continuous variable is non-negative and upper-bounded by a constant. Both is the case (note that we generally only model a bounded cardinality range which implies also an upper bound on the number of loop iterations). The advantage of the second representation is that we only need to introduce a number of variables and constraints that is linear in the number of tables (instead of linear in the number of thresholds like for the first possibility).

We have seen that join orders, the cardinality of intermediate results, and the cost of join operations according to standard cost formulas can all be represented as MILP. In the next section we introduce several extensions of the problem model that we used so far.

%We assume in this subsection that the cost of the entire query plan can be written as linear function in the cardinality of input and output of each operation (such functions constitute in fact an important sub-class of cost functions that are representative for other standard cost functions as well~\cite{Cluet1995}). We make the linear cost function the objective and tasks the solver with minimizing it.

\section{Extensions}
\label{extensionsSec}

%TODO: correlated predicates, cross products not allowed

We introduced our mapping for query plans by means of a basic problem model that focuses on join order. We discuss extensions of the query language, of the query plan model, and of the cost model in this section. %Due to space restrictions, we only sketch out how the extensions can be implemented. 

Note that not all proposed extensions are necessary in each scenario: the basic model introduced in the last section allows for instance to find join orders which minimize the sum of intermediate result sizes. Such join orders are optimal according to many standard operator cost functions~\cite{Cluet1995}. It is therefore in many scenarios possible to obtain good query plans based on the join order that was calculated using the basic model. To transform a join order into a query plan, we choose optimal operator implementations based on the cardinality of the join operands, we evaluate predicates as early as possible (predicate push-down), and we project out columns as soon as they are not required anymore. 

An alternative is to let the MILP solver make some of the decisions related to projection, predicate evaluation, and join operator selection. We show how this can be accomplished if desired. In addition, we discuss extensions of the cost and query model. 

In Section~\ref{predicatesSub}, we discuss how to represent n-ary predicates, correlated predicates, and predicates that are expensive to evaluate. We show how to handle projections in Section~\ref{projectionSub} and in Section~\ref{implementationsSub} we show how the MILP solver can choose between different operator implementations. We show how to handle interesting orders and other intermediate result properties in Section~\ref{propertiesSub}. In Section~\ref{languageExtensions}, we finally discuss how we can extend our approach to handle queries with aggregates and nested queries.

We sketch out the following extensions relatively quickly due to space restrictions. They use however similar ideas as we applied in the last section. Our goal is less to provide a detailed model for each possible scenario but rather to demonstrate that the MILP formalism is flexible enough to cover the most relevant aspects of query optimization.

%practical to In many scenarios, good query plans can be derived from a join order minimizing intermediate result set sizes by choosing optimal join operators for the given join order, evaluating predicates as early as possible (predicate push-down), and projecting out columns as soon as they are not required anymore. Nevertheless, we introduce techniques by which the MILP optimizer can for instance select join operator implementations and choose when to use projection and when to evaluate predicates. 

\subsection{Predicate Extensions}
\label{predicatesSub}

So far we have considered binary predicates. We show how n-ary predicates can be modeled. Let $p$ be an n-ary predicate. N-ary predicates refer to n tables and we designate by $T_1(p)$ to $T_n(p)$ the tables on which $p$ is evaluated. All tables that $p$ refers to must be present in the operands in which $p$ is evaluated. If $pao_{pj}$ indicates whether predicate $p$ can be evaluated in the outer operand of the $j$-th join then we must introduce constraint $pao_{pj}\leq tio_{T_i(p)j}$ for each join and each $i\in\{1,\ldots,n\}$. This forces variables $pao_{pj}$ to zero if at least one table is not present. Note that we must introduce analogue predicate variables for the inner operands for all unary predicates.

In our basic model, we assume that predicates are uncorrelated. Then the accumulated selectivity of a predicate group corresponds always to the product of the selectivity values of the single tables. In reality this is not always the case, even if it is a common simplification to assume uncorrelated predicates. Assume that there is a correlated group $P_{cor}$ of predicates such that the accumulated selectivity of all predicates in $P_{cor}$ differs significantly from their selectivity product. Then we introduce a new predicate $g$ that represents the correlated predicate group. The selectivity $Sel(g)$ is chosen in a way such that $Sel(g)\prod_{p\in P_{cor}}Sel(p)$ yields the correct selectivity, taking correlations into account. So the selectivity of $g$ \textit{corrects} the erroneous selectivity that is based on the assumption of independent predicates. 

Now we just need to make sure that the predicate variable associated with $g$ is set to one in all operands in which all predicates from $P_{cor}$ are selected but not otherwise. We force $pao_{gj}$ to one if all correlated predicates are present by requiring $pao_{gj}\geq 1-|P_{cor}|+\sum_{p\in P_{cor}}pao_{pj}$. We force $pao_{gj}$ to zero if at least one of the correlated predicates is not activated by introducing $n$ constraints of the form $pao_{gj}\leq pao_{pj}$ for $p\in P_{cor}$. No other constraints need to be introduced for $pao_{gj}$ but terms including $pao_{gj}$ must be included in all expressions representing cardinality, byte size, etc.

So far we have assumed that predicate evaluation is not associated with cost. We constrained the variables $pao_{pj}$ only to zero if required tables are not in the operand. We did not explicitly force them to one at any point since, as they reduce cardinality, their evaluation reduces cost and the MILP solver will generally choose to evaluate them as early as possible. 

This model is not always appropriate. If predicate evaluations are expensive then it can be preferable to postpone their evaluation~\cite{Chaudhuri1999, Hellerstein1993, Kemper1994}. The predicate-related variables $pao_{pj}$ influence the cardinality estimates of join operands. They capture whether the corresponding predicate \textit{was} already evaluated as otherwise it cannot influence cardinality. We cannot use those variables directly to incorporate the cost of predicate evaluations. The effect on cardinality of having evaluated a predicate once will persist for all future operations. The evaluation cost needs however only to be payed once. We introduce additional variables $pco_{pj}$ (short for \textit{Predicate evaluation Cost for Outer operand}) and set $pco_{pj}=pao_{p,j+1}-pao_{p,j}$. Intuitively, the predicate was evaluated in the current join if it is evaluated in the input to the next join but not in the input of the current join. The sum $\sum_{j}pco_{pj}co_j$ yields the evaluation cost associated with predicate $p$ (we can additionally weight by a factor that represents predicate evaluation cost per tuple). This is not a linear function as we multiply variables. We have however a product between a binary variable and a continuous variable again. As before, we can transform such expressions into a set of linear constraints and a new auxiliary variable~\cite{Modeling2015}.

Now that evaluation of predicates is not automatically desirable anymore, we must introduce additional constraints making sure that all predicates are evaluated at the end of query execution. Designating by $j_{max}$ the index of the last join, we simply set $pao_{p,j_{max}+1}=1$ by convention. This means that each predicate that was not evaluated before the last join must be evaluated during the last join since $pco_{pj_{max}}=1-pao_{pj_{max}}$. We finally introduce constraints making sure that no predicate is initially evaluated and we introduce constraints making sure that an evaluated predicate remains evaluated. The latter constraints are in fact optional since additional predicate evaluations increase the cost. Depending on the solver implementation, it can nevertheless be beneficial to add such constraints to reduce the search space size.

\subsection{Projection}
\label{projectionSub}

Our cost formulas have so far been based on cardinality alone as we have assumed a constant byte size per tuple. This is of course a simplification and we must in general take into account the columns that we project on and their byte sizes. We designate by $L$ the set of columns over all query tables. By $Byte(l)$ we denote the number of bytes per tuple that column $l\in L$ requires. We introduce one variable $clo_{jl}$ (short for \textit{CoLumn in Outer operand}) for each join~$j$ and each column $l\in L$ to indicate whether column~$l$ is present in the outer operand of join $j$ (and analogue variables for the inner operands). Then a refined formula for the estimated number of bytes consumed by the outer operand is $co_j\cdot\sum_{l\in L}clo_{jl}Byte(l)$. This is the sum over products between a constant ($Byte(l)$), a binary variable ($clo_{jl}$), and a continuous variable that takes only non-negative values ($co_j$). This formula can be expressed using only linear constraints using the same transformations that we used already before~\cite{Modeling2015}. Special rules apply for the inner operand again: for the inner operand, we can estimate the byte size (or any derived measure such as the number of disc pages) by summing over the column variables, weighted by the column byte size as well as by the cardinality of the table that the column belongs to.

We must still constrain the variables $clo_{jl}$ to make sure that only valid query plans can be represented. First of all we must connect columns to their respective tables. If the table associated with a column is not present then the column cannot be present either in a given operand. If column $l$ is associated with table $t$ then the constraint $clo_{jl}\leq tio_{tj}$ forces the column variable to zero if the associated table is not present. Not selecting any columns would be the most convenient way for the optimizer to reduce plan costs. To prevent this from happening, we must enforce that all columns that the query refers to are in the final result. Also, we must enforce that all columns that predicates refer to are present once they are evaluated. We introduced variables indicating the immediate evaluation of a predicate during a specific join. Those are the variables that need to be connected to the columns they require via corresponding constraints. We must also make sure that a column cannot reappear in later joins after it has been projected out (otherwise that would be a convenient way of reducing intermediate result sizes while still satisfying the constraints requiring certain columns in the final result). Introducing constraints of the form $clo_{jl}\geq clo_{j+1,l}$ satisfies that requirement.

\subsection{Choosing Operator Implementations}
\label{implementationsSub}

We have already discussed the cost functions of different join operator implementations in the last section. So far we have however assumed that only one of those cost functions is used to calculate the cost for all joins. This allows to select optimal operator implementations after a good join order, minimizing intermediate result sizes, has been found. We can however also task the MILP solver to pick operator implementations as we outline in the following.

Denote by $I$ the set of join operator implementations. We have shown how to calculate join cost for each of the standard join operators. We can introduce a variable $pjc_{ji}$ (short for \textit{Potential Join Cost}) for each join~$j$ and for each operator implementation $i\in I$ representing the cost of the join if that operator is used. We use the term \textit{potential} since whether that cost is actually counted depends on whether or not the corresponding operator implementation is selected. 

We introduce binary variables $jos_{ji}$ (short for \textit{Join Operator Selected}) to indicate for each operator implementation $i$ and join $j$ whether the operator was used to realize the join. We require that exactly one implementation is selected for each join as expressed by the constraint $\sum_{i}jos_{ji}=1$ that we must introduce for each join. Having the potential cost for each join operator as well as information on which operator is selected, we can for each operator calculate the actual join cost $ajc_{ji}$. The actual join cost associated with one specific operator implementation is \textit{zero} if that operator is not selected. Otherwise (if that operator is selected) the actual join cost corresponds to the potential join cost. We have the following relationship between potential and actual join cost $ajc_{ji}=jos_{ji}\cdot pjc_{ji}$. Here we multiply a binary with a non-negative continuous variable and can apply the same linearization as before~\cite{Modeling2015}. The sum over the actual join cost variables over all operator implementations yields the cost of each join operation.

%The plan space that we considered so far centers on join order but not on the selection of join operator implementations. Having a join order that minimizes the size of intermediate results minimizes the cost for many standard operators~\cite{Cluet1995} but we additionally show the choice between operator implementations can be incorporated into the MILP model. For each join or scan operation, we introduce a set of binary variables, one variable for each available operator implementation. We require that the sum over the variables representing alternative implementations is one for each operation such that exactly one implementation is chosen. Then we can represent the cost of an operation as a weighted sum over the  implementation variables, weighted by the cost of each implementation for the corresponding operation. The cost of each implementation is a variable as it depends on the cardinality or byte size values. We have therefore a sum of products between binary variables (representing operator implementation choices) and continuous variables (representing operand size). As before, we can linearize such terms by standard techniques~\cite{Modeling2015}. 

\subsection{Intermediate Result Properties}
\label{propertiesSub}

Alternative join operator implementations can sometimes produce intermediate results with different physical properties (while the contained data remains the same over all alternative implementations). Tuple orderings are perhaps the most famous example~\cite{Selinger1979}. If tuples are produced in an interesting order then the cost of successive operations can be reduced (e.g., the sorting stage can be dropped for a sort-merge join). Also, the distinction whether an intermediate result is written to disc or remains in main memory is a physical property of that result and influences the cost of successive operations.

Assume that we consider a set $X$ of relevant intermediate result properties. Then we can introduce a binary variable $ohp_{jx}$ (short for \textit{Outer operand Has Property}) indicating whether the outer operand of the $j$-th join has property~$x$. Property $x$ could for instance represent the fact that the corresponding result is materialized. Property~$x$ could also represent one specific tuple ordering.

The properties constrain the operator implementations that can be applied for the next join. We could for instance introduce one operator implementation representing a pipelined block nested loop join while another operator implementation represents a block nested loop join without pipelining. The applicability of the pipelined join would have to be restricted based on whether or not the corresponding input remains in memory. If implementation $i$ requires property $x$ in the outer join operand in order to become applicable then we can impose the constraint $jos_{ji}\leq ohp_{jx}$ to express that fact.

Operators such as the sort-merge join can be decomposed into different sub-operators (e.g., sorting the outer operand, sorting the inner operand, merging). This avoids having to introduce a new variable for each possible combination of situations (e.g., outer operand sorted and inner operand sorted, outer operand sorted but inner operand not sorted, etc.). 

Whether an intermediate result has a certain physical property is determined by the operator which produces the result (and possibly by properties of the input to the producing operation). If a subset $\widetilde{I}\subseteq I$ produces results with a certain property $x$ then we can set $ohp_{j+1,x}=\sum_{i\in\widetilde{I}}jos_{ij}$. As only one of the operators is selected, the aforementioned constraint is valid and sets the left expression either to zero or to one. Certain properties such as interesting orders might be provided automatically by certain tables (if the data on disk has that order). Then we need additional constraints to connect properties to tables.

In summary, we have shown that all of the most important aspects of query optimization can be represented in the MILP formalism.

\subsection{Extended Query Languages}
\label{languageExtensions}

We have already implicitly discussed several extensions to the query language in this section. We discussed how non-binary predicates and projection are supported. This gives us a system handling select-project-join (SPJ) queries. 

It is generally common to introduce query optimization algorithms using SPJ queries for illustration. There are however standard techniques by which an optimization algorithm treating SPJ queries can be extended into an algorithm handling richer query languages. 

The seminal paper by Selinger~\cite{Selinger1979} describes how a complex SQL statement containing nested queries can be decomposed into several simple query blocks that use only selection, projection, and joins; the join order optimization algorithm is applied to each query block separately. Later, the problem of unnesting a complex SQL statement containing aggregates and sub-queries into simple SPJ blocks has been treated as a research problem on its own; corresponding publications focus on the unnesting algorithms and use join order optimization algorithms as a sub-function (e.g., \cite{Muralikrishna1992}).
\section{Formal Analysis}
\label{analysisSec}

State-of-the art MILP solvers use a plethora of heuristics and optimization algorithms which makes it hard to predict the run time for a given MILP instance. It is however a reasonable assumption that optimization time tends to increase in the number of variables and constraints, even if preprocessing steps are sometimes able to eliminate redundant elements. The assumptions that we make here are supported by the experimental results that we present in the next section: we see a strong (even if not perfect) correlation between the number of variables and constraints and the MILP solver performance. 

For the aforementioned reasons, we study in the following how the asymptotic number of variables and constraints in the MILP grows in the dimensions of the query optimization problem from which it was derived. We denote in the following by $n=|Q|$ the number of query tables to join and by $m=|P|$ the number of predicates. By $l=|\Theta|$ we denote the number of thresholds that are used to approximate cardinality values. The following theorems refer to the basic problem model that was presented in Section~\ref{approachSec}. %We discuss the effect of the extensions described in Section~\ref{extensionsSub} afterwards. 

\begin{theorem}
The MILP has $O(n\cdot(n+m+l))$ variables.
\end{theorem}
\begin{proof}
Give $n$ tables to join, each complete query plan has $O(n)$ joins. We require $O(n)$ binary variables per join to indicate which tables form the join operands, we require $O(m)$ binary variables per operand to indicate which predicates can be evaluated, and we require $O(l)$ continuous variables per operand to calculate cardinality estimates. 
\end{proof}

\begin{theorem}
The MILP has $O(n\cdot(n+m+l))$ constraints.
\end{theorem}
\begin{proof}
For each join operand we need $O(n)$ constraints to restrict table selections, $O(m)$ constraints to restrict predicate applicability, and $O(l)$ constraints to force the threshold variables to the right value.
\end{proof}

% Plan space, metric, join graph structure
\def\addAllAnalysisPlots#1#2#3{
\addplot+[unbounded coords=jump,bar shift=-0.1cm, bar width=0.1cm, draw=black, fill=blue] table[header=true,col sep=tab, x index=0, y index=4] {plotsdata/#1/#2_MEDIAN_#3_MN};
\addplot+[unbounded coords=jump,bar shift=0cm, bar width=0.1cm, fill=brown, draw=black] table[header=true,col sep=tab, x index=0, y index=3] {plotsdata/#1/#2_MEDIAN_#3_MN};
\addplot+[unbounded coords=jump,bar shift=0.1cm, bar width=0.1cm, fill=red, draw=black] table[header=true,col sep=tab, x index=0, y index=2] {plotsdata/#1/#2_MEDIAN_#3_MN};
}

% column, row, text
\def\plotTitle#1#2#3{
\node at (group c#1r#2.north) [yshift=0.15cm] {#3};
}

\newlength{\comparisonPlotHeight}
\setlength{\comparisonPlotHeight}{3.5cm}

%% feature title, feature name, ymode
%\def\comparisonGroupPlot#1#2#3{

\section{Experimental Evaluation}
\label{experimentsSec}

%TODO: have to treat solvers as black boxes

Using existing MILP solvers as base for the query optimizer reduces coding overhead and automatically yields parallelized anytime query optimization due to the features of typical MILP solvers. In this section, we compare the performance of a MILP based optimizer to a classical dynamic programming based query optimization algorithm. 

We describe and justify our experimental setup in Section~\ref{setupSub} and discuss our results in Section~\ref{resultsSub}.

\subsection{Experimental Setup}
\label{setupSub}

\begin{figure}
\centering
\begin{tikzpicture}
\begin{groupplot}[group style={group size=2 by 1, x descriptions at=edge bottom}, width=4.25cm, height=\comparisonPlotHeight,
xlabel=Nr.\ query tables, xlabel near ticks, ylabel near ticks, 
ymode=linear, ybar, %xmin=6, xmax=60,
legend to name=comparisonLeg, legend columns=3,
xtick=data, ymajorgrids, ylabel style={font=\small},
enlargelimits=0.2, max space between ticks=10]
% first row
\nextgroupplot[ylabel=Nr.\ Variables, height=\comparisonPlotHeight]
\addAllAnalysisPlots{linearAll}{numVars}{STAR}
% second row
\nextgroupplot[legend to name=analysisLegend, legend columns=2, ylabel=Nr.\ Constraints, height=\comparisonPlotHeight]
\addAllAnalysisPlots{linearAll}{numConstr}{STAR}
\addlegendentry{ILP (Low Precision)}
\addlegendentry{ILP (Medium Precision)}
\addlegendentry{ILP (High Precision)}
\end{groupplot}
\end{tikzpicture}

\ref{analysisLegend}
\caption{Median number of variables and constraints of a MILP problem representing the optimization of one query.\label{analysisFig}}
\end{figure}
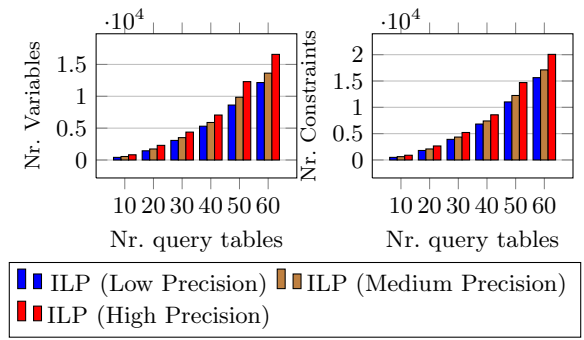

We implemented a prototype of the MILP based optimizer that was introduced in the last sections. We transform query optimization problems into MILP problems and use the Gurobi\footnote{\url{http://www.gurobi.com/}} solver in version 5.6.3 to find optimal or near-optimal solutions to the resulting MILP problems. The MILP solution is read out and used to construct a corresponding query plan. %Note that it seems more appropriate to prototype our approach within a freshly written optimizer than to integrate it into an existing database system: certain advantages that we claim for our approach, namely the reduced implementation overhead, only come into play in the first scenario.

We compare this approach against the classical dynamic programming algorithm by Selinger~\cite{Selinger1979}. Dynamic programming algorithms are very popular for exhaustive query optimization~\cite{Moerkotte2006, Moerkotte2008} and are for instance used inside the optimizer of the Postgres database system\footnote{\url{http://www.postgresql.org/}}. 

We compare the two aforementioned algorithms on randomly generated queries. We generate queries according to the method proposed by Steinbrunn et al.~\cite{Steinbrunn1997} which is widely used to benchmark query optimization algorithms~\cite{Steinbrunn1997, Bruno, Trummer2015}. We generate queries of different sizes (referring to the number of tables to join) and with different join graph structures (chain graphs, star graphs, and cycle graphs~\cite{Steinbrunn1997}). We allow cross products which increases the search space size significantly compared to the case without cross products~\cite{Ono1990}. 

%We use a simple cost model to estimate the cost of query plans that is based on the sum of intermediate result sizes. Minimizing cost according to that model was however shown to minimize cost according to the cost models for other standard join operators as well~\cite{Cluet1995}. We could extend our prototype to use different cost models as well as discussed in Section~\ref{extensionsSub}. 

We assume that hash joins are used and search the optimal join order. The MILP approach approximates the byte sizes of the intermediate results and therefore the cost of join operations. We evaluate three configurations of our algorithm that differ in the precision by which they approximate cardinality (higher approximation precision requires more MILP variables and constraints). Our first configuration offers high precision and approximates cardinality with a tolerance of factor 3. Our second configuration reduces approximation precision and has a tolerance factor of 10. Our third configuration reduces approximation precision further and has tolerance factor 100. Our most precise configuration uses 60 threshold variables per intermediate result up to 40 table joins and 100 threshold variables per result for queries joining 50 and 60 tables. At the other side of the spectrum is the low-precision configuration which uses 15 threshold variables per result for up to 40 tables and 25 variables for more than 40 tables.

We compare algorithms by the quality of the plans that they produce after a certain amount of optimization time. We allow up to 60 seconds of optimization time and compare the output generated by all algorithms in regular time intervals. The high amount of optimization time seems justified since we compare the algorithms also on very large queries. All compared algorithms need significantly less time than 60 seconds to produce optimal plans for small queries. Investing 60 seconds into optimization can however be well justified if queries are executed on big data where choosing a sub-optimal plan can have devastating consequences~\cite{Soliman2014}. 

During the 60 seconds of optimization time, we compare optimization algorithms in regular intervals according to the following criterion. We compare them based on the factor by which the cost of the best plan found so far is higher than the optimum at most. MILP solvers calculate such bounds based on the integrality gap. The classical dynamic programming algorithm is not an anytime algorithm but after its execution finishes, the produced plan is optimal and hence the optimality factor is one. 

We \textit{do not} compare algorithms based on the cost overhead that the generated plans have compared to the optimum. Instead, we compare them based on an \textit{upper bound} on the relative cost overhead that the algorithm can formally guarantee at a certain point in time. The actual cost overhead is only known in hindsight after optimization has finished (and for some of the query sizes we consider, calculating the truly optimal query plans would cause high computational overheads). The upper bound that we use as criterion is the only value that is known at optimization time and therefore the only value on which termination decisions can be based on for instance (e.g., we could terminate optimization once the query optimizer is certain that the current plan is not more expensive than the optimum by more than factor 2). 

The comparison criterion that we use excludes any randomized or heuristic query optimization algorithms~\cite{Bennett1991, Bruno, ioannidis1990randomized, Steinbrunn1997, Swami1988, Swami1989} from our experimental evaluation: such algorithms cannot give any formal guarantees on the optimality of the produced plans. They cannot even give upper bounds on the relative cost overhead of the generated plans.

Our algorithms (for the MILP approach: the part that transforms query optimization into MILP) are implemented in Java~1.7. The experiments were executed using the Java HotSpot(TM) 64-Bit Server Virtual Machine version on an iMac with  i5-3470S 2.90GHz CPU and 16~GB of DDR3 RAM. 

% Plan space, metric, join graph structure, tables
\def\addAllPerformancePlots#1#2#3#4{
\addplot+[unbounded coords=jump,mark=triangle,mark size=3,draw=black] table[header=true,col sep=tab, x expr=\thisrowno{0}*6+6, y index=1] {plotsdata/#1/#2_MEDIAN_#3_MN_T#4};
\addplot+[unbounded coords=jump,mark=o,mark size=1.5,draw=red] table[header=true,col sep=tab, x expr=\thisrowno{0}*6+6, y index=2] {plotsdata/#1/#2_MEDIAN_#3_MN_T#4};
\addplot+[unbounded coords=jump,mark=square,mark size=1.5,draw=brown] table[header=true,col sep=tab, x expr=\thisrowno{0}*6+6, y index=3] {plotsdata/#1/#2_MEDIAN_#3_MN_T#4};
\addplot+[unbounded coords=jump,mark=asterisk,mark size=2,draw=blue] table[header=true,col sep=tab, x expr=\thisrowno{0}*6+6, y index=4] {plotsdata/#1/#2_MEDIAN_#3_MN_T#4};
}

% nr metrics, feature name, graph type
\def\comparisonPlot#1#2#3{
\addplot+[unbounded coords=jump, bar shift=-0.1cm, bar width=0.1cm] table[header=true,col sep=tab, x index=0, y index=7] {plotsdata/furtherAnalysis/#1M#3#2.txt};
\addlegendentry{NSGA-II}
\addplot+[unbounded coords=jump, bar shift=0cm, bar width=0.1cm] table[header=true,col sep=tab, x index=0, y index=8] {plotsdata/furtherAnalysis/#1M#3#2.txt};
\addlegendentry{II}
\addplot+[unbounded coords=jump, bar shift=0.1cm, bar width=0.1cm] table[header=true,col sep=tab, x index=0, y index=9] {plotsdata/furtherAnalysis/#1M#3#2.txt};
\addlegendentry{RMQ}
}

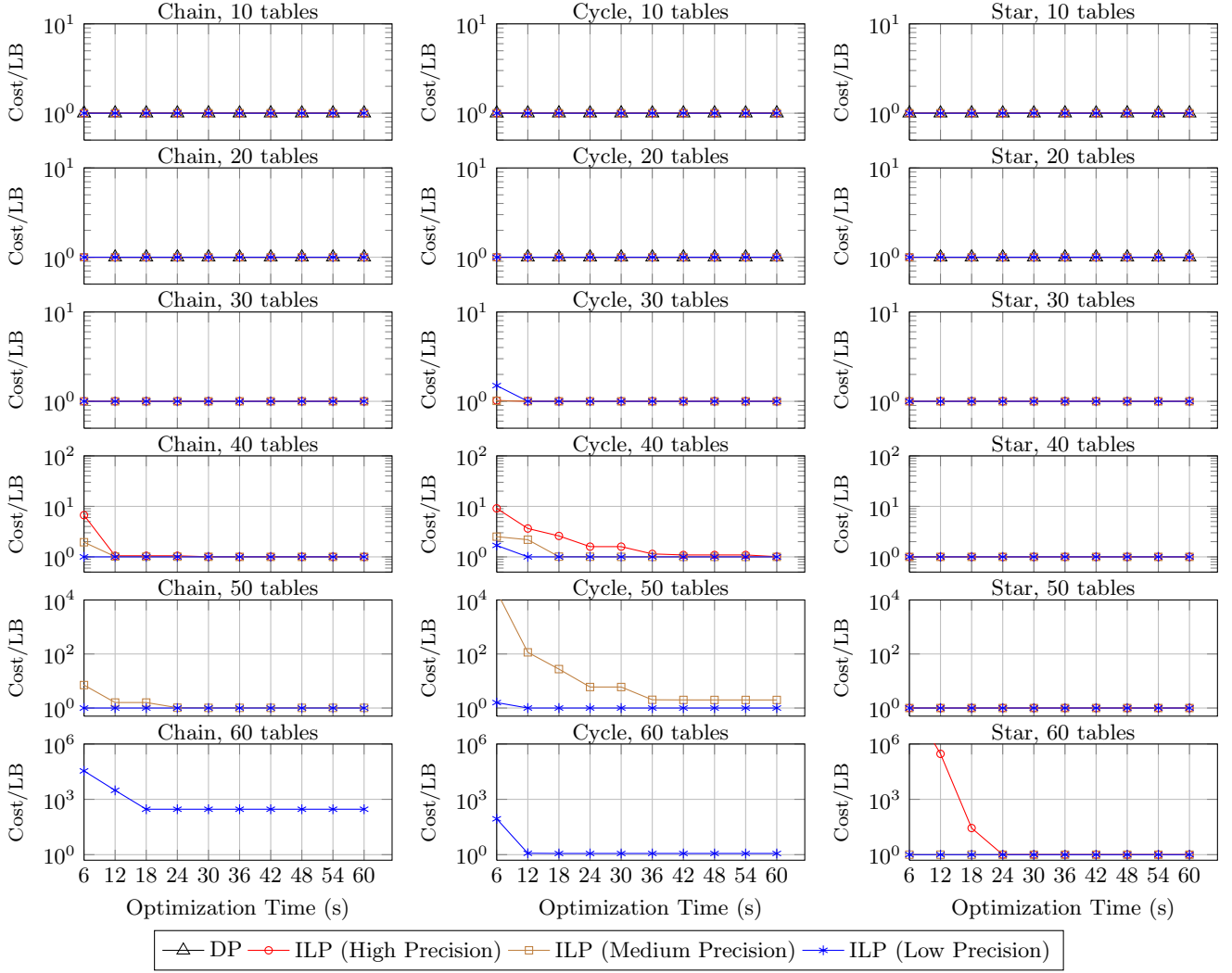
\begin{figure*}[t!]
\centering
\begin{tikzpicture}[declare function={Infinity=inf;}]
\begin{groupplot}[group style={group size=3 by 6, x descriptions at=edge bottom, horizontal sep=1.5cm, vertical sep=0.4cm}, 
width=6cm, height=3.25cm,
xlabel=Optimization Time (s), ylabel=Cost/LB, xlabel near ticks, ylabel near ticks,
ymode=log, xmajorgrids, xmin=6, ymin=1, ymax=10E3,%xmax=2.9, 
ymajorgrids, xtick={6,12,18,24,30,36,42,48,54,60},ylabel style={font=\small},
ymin=0.5]
% 10 tables
\nextgroupplot[ymax=10]
\addAllPerformancePlots{linearAll}{costGap}{CHAIN}{10}
\nextgroupplot[ymax=10]
\addAllPerformancePlots{linearAll}{costGap}{CYCLE}{10}
\nextgroupplot[ymax=10]
\addAllPerformancePlots{linearAll}{costGap}{STAR}{10}
% 20 tables
\nextgroupplot[ymax=10]
\addAllPerformancePlots{linearAll}{costGap}{CHAIN}{20}
\nextgroupplot[ymax=10]
\addAllPerformancePlots{linearAll}{costGap}{CYCLE}{20}
\nextgroupplot[ymax=10]
\addAllPerformancePlots{linearAll}{costGap}{STAR}{20}
% 30 tables
\nextgroupplot[ymax=10]
\addAllPerformancePlots{linearAll}{costGap}{CHAIN}{30}
\nextgroupplot[ymax=10]
\addAllPerformancePlots{linearAll}{costGap}{CYCLE}{30}
\nextgroupplot[ymax=10]
\addAllPerformancePlots{linearAll}{costGap}{STAR}{30}
% 40 tables
\nextgroupplot[ymax=100]
\addAllPerformancePlots{linearAll}{costGap}{CHAIN}{40}
\nextgroupplot[ymax=100]
\addAllPerformancePlots{linearAll}{costGap}{CYCLE}{40}
\nextgroupplot[ymax=100]
\addAllPerformancePlots{linearAll}{costGap}{STAR}{40}
% 50 tables
\nextgroupplot[ymax=10000]
\addAllPerformancePlots{linearAll}{costGap}{CHAIN}{50}
\nextgroupplot[ymax=10000]
\addAllPerformancePlots{linearAll}{costGap}{CYCLE}{50}
\nextgroupplot[ymax=10000]
\addAllPerformancePlots{linearAll}{costGap}{STAR}{50}
% 60 tables
\nextgroupplot[ymax=1000000]
\addAllPerformancePlots{linearAll}{costGap}{CHAIN}{60}
\nextgroupplot[ymax=1000000]
\addAllPerformancePlots{linearAll}{costGap}{CYCLE}{60}
\nextgroupplot[ymax=1000000,legend to name=thisPlotLegend, legend columns=4]
\addAllPerformancePlots{linearAll}{costGap}{STAR}{60}
% legend
\addlegendentry{DP}
\addlegendentry{ILP (High Precision)}
\addlegendentry{ILP (Medium Precision)}
\addlegendentry{ILP (Low Precision)}
\end{groupplot}

\plotTitle{1}{1}{Chain, 10 tables}
\plotTitle{1}{2}{Chain, 20 tables}
\plotTitle{1}{3}{Chain, 30 tables}
\plotTitle{1}{4}{Chain, 40 tables}
\plotTitle{1}{5}{Chain, 50 tables}
\plotTitle{1}{6}{Chain, 60 tables}

\plotTitle{2}{1}{Cycle, 10 tables}
\plotTitle{2}{2}{Cycle, 20 tables}
\plotTitle{2}{3}{Cycle, 30 tables}
\plotTitle{2}{4}{Cycle, 40 tables}
\plotTitle{2}{5}{Cycle, 50 tables}
\plotTitle{2}{6}{Cycle, 60 tables}

\plotTitle{3}{1}{Star, 10 tables}
\plotTitle{3}{2}{Star, 20 tables}
\plotTitle{3}{3}{Star, 30 tables}
\plotTitle{3}{4}{Star, 40 tables}
\plotTitle{3}{5}{Star, 50 tables}
\plotTitle{3}{6}{Star, 60 tables}
\end{tikzpicture}
\ref{thisPlotLegend}
\caption{Comparing dynamic programming based optimizer versus integer linear programming for left-deep query plans.\label{leftPerformanceFig}}
\end{figure*}

\subsection{Experimental Results}
\label{resultsSub}

%Each data point is the median out of 20 randomly generated test cases. 

We start by analyzing the size of the generated MILP problems. Figure~\ref{analysisFig} shows the number of constraints and variables. We show results for queries with a star-shaped join graph structure while the results for chain and cycle graph structures differ only marginally (the only difference is that cycle graphs require one additional predicate variable per intermediate result compared to star graphs). The ILP configuration with higher approximation precision requires in all cases more variables and constraints. For all configurations, the number of variables and constraints increases with increasing number of query tables. 

Figure~\ref{leftPerformanceFig} shows performance results for left-deep plans. We allow cross product joins. The experimental setup was explained and justified in Section~\ref{setupSub}. The figure shows median values for 20 randomly generated queries. For 10 query tables, all compared algorithms find the optimal plan very quickly. For 20 query tables, the dynamic programming approach already takes more than six seconds in average to find the optimal plan while the MILP approach is faster. With 20 query tables we are reaching the limit of what is usually considered practical by dynamic programming algorithms. Also note that we allow cross product joins which increases the size of the plan space significantly. 

For higher numbers of query tables, up to 60, the dynamic programming approach does not return any plan within one minute of optimization time. Note that increasing the number of tables by 10 increases the number of table sets that the dynamic programming approach must consider by factor $2^{10}=1024$. It is therefore not surprising that this algorithm is not able to optimize queries with 30 tables and more.

All configurations of the MILP approach find optimal or at least guaranteed near-optimal plans for up to 40 tables, often already after a few seconds. For 50 and 60 table joins, all MILP configurations are able to find plans quickly for star join graphs. For cycle graphs, the low-precision configuration finds still optimal plans up to 60 tables while the medium-precision configuration finds near-optimal plans. Both configurations find optimal plans for 50 tables and chain graphs while this is not possible for queries with 60 tables and a chain graph structure. This means that optimization of chain and cycle queries seems to be more challenging for MILP approaches than optimization of star queries. Note that star queries are more difficult to optimize when excluding cross products and applying dynamic programming~\cite{Ono1990}; for MILP approaches it is apparently the opposite.

%This surprised us since the number of variables and constraints increases by a non-negligible factor when switching from low to high approximation precision (see Figure~\ref{analysisFig}). Only for 40 tables and cyclic join graphs is the difference noticeable where the upper bound on the cost factor is 9.1 after the first optimization time period for the high-precision configuration while it is 1.7 for the low-precision variant. For small tables there are no noticeable performance differences across different join graph structures. For higher number of tables, cyclic join graph structures turn out to be the most difficult ones for the MILP approach. 

%Comparing the three configurations, we often see faster results from the low-precision configuration while the opposite case occurs as well. MILP solvers use a plethora of preprocessing steps, heuristics, and search algorithms and selected automatically the best methods to apply for a given problem; different queries might very well have been optimized by different algorithms inside the same solver. It is therefore not surprising that the expected correlation between optimization time and the number of MILP variables and constraints is not perfect while we see it at a grand scale.

We conclude that the MILP approach does not only match but even outperforms traditional exhaustive query optimization algorithms for left-deep plan spaces by a significant margin.

\section{Conclusion}
\label{conclusionSec}

Basing newly developed query optimizers on existing MILP solver implementations reduces the size of the optimizer code base and allows to benefit from features such as parallelization and anytime behavior that those solvers encapsulate. %We have shown how to transform the join ordering problem into MILP. 

We have demonstrated how to transform query optimization into MILP. Our experimental results show that MILP approaches can outperform traditional dynamic programming approaches significantly. 

Generally it should be noted that the experimental results in this paper are only snapshots and not intrinsic to the proposed mapping: as new MILP solver generations appear, the performance of our MILP based approach is likely to improve further without having to adapt the mappings.

\section{Acknowledgment}

This work was supported by ERC Grant 279804 and by a European Google PhD fellowship.

\bibliographystyle{abbrv}

\balance

\end{document}